\documentclass[a4paper]{llncs}
\usepackage{xcolor,amssymb,amsmath,xspace,suffix,tikz,algorithm,algorithmic,picinpar,verbatim,subfigure,pgfplots,lmodern}
\usetikzlibrary{automata,arrows,shapes,decorations.pathmorphing,matrix}
\usepackage[T1]{fontenc}
\usepackage[fancyhdr]{svninfo}
\svnInfo $Id: paper.tex 1334 2010-09-03 09:52:49Z scranen $

\newcommand{\SC}[2][]{{\color{blue!80!cyan}#2}}
\newcommand{\TW}[1]{{\color{red!80!blue}#1}}
\newcommand{\JK}[1]{{\color{green!50!black}#1}}
\WithSuffix\newcommand\SC*[1]{\marginpar{\color{blue!80!cyan}#1}} 
\WithSuffix\newcommand\TW*[1]{\marginpar{\color{red!80!blue}#1}} 
\WithSuffix\newcommand\JK*[1]{\marginpar{\color{green!50!black}#1}} 
\newcommand{\ie}{\textit{i.e.}\xspace} 
\newcommand{\eg}{\textit{e.g.}\xspace}
\newcommand{\viz}{\textit{viz.}\xspace}

\mathchardef\mhyphen="2D 
\newcommand{\rel}[1]{\mathrel{#1}}

\newcommand{\N}{\ensuremath{\mathbb{N}}\xspace}
\newcommand{\R}{\ensuremath{\rel{R}}}

\newcommand{\game}{\ensuremath{\mathcal{G}}\xspace}
\newcommand{\even}{\ensuremath{0}\xspace}
\newcommand{\odd}{\ensuremath{1}\xspace}
\newcommand{\player}{\ensuremath{\mathit{i}}\xspace}
\newcommand{\priority}{\ensuremath{\Omega}\xspace}
\newcommand{\getplayername}{\ensuremath{\mathcal{P}}}
\newcommand{\getplayer}[1]{\getplayername(#1)}

\newcommand{\prefixes}[3][]{\ensuremath{\Pi_{#2}^{#1}(#3)}}

\newcommand{\consistency}{\ensuremath{\Vdash}}
\newcommand{\consistent}[2]{\ensuremath{{#2}\consistency{#1}}}
\newcommand{\path}[1]{\ensuremath{\langle #1 \rangle}}
\newcommand{\pathprefix}[3]{\ensuremath{#1_{#2}^{#3}}}
\newcommand{\To}{\ensuremath{\Longrightarrow}}

\newcommand{\mimick}{\mathsf{mimick}_{\phi,v}}
\newcommand{\winner}{\ensuremath{\sim_w}}
\renewcommand{\sb}{\sim}

\newcommand{\divr}[2]{\mathsf{div}_{#1}(#2)}
\newcommand{\stut}{\eqsim}

\newcommand{\entry}[2]{\mathsf{reach}_{#1}(#2)}
\newcommand{\vertexorder}{\ensuremath{\sqsubset}}
\newcommand{\targetorder}[1]{\ensuremath{\prec_{#1}}}

\newcommand{\vertexordermin}{\sqcap}
\newcommand{\targetordermin}[1]{\curlywedge_{#1}}

\newcommand{\pathconcat}{\cdot}
\newcommand{\targetclass}[2]{\mathsf{targetclass}_{#1}(#2)}
\newcommand{\target}[2]{\tau_{#1}(#2)}

\newcommand{\dom}[1]{\mathsf{dom}(#1)}

\begin{document}

\tikzstyle{even}=[shape=diamond,inner sep=1pt, minimum size=8pt]
\tikzstyle{odd}=[shape=rectangle,inner sep=1pt, minimum size=8pt]  

\title{Stuttering Equivalence for Parity Games}
\author{Sjoerd Cranen \and Jeroen J.A. Keiren \and Tim A.C. Willemse}
\institute{Department of Mathematics and Computer Science,\\
    Technische Universiteit Eindhoven, \\
    P.O. Box 513, 5600 MB Eindhoven, The Netherlands}

\maketitle
\begin{abstract} We study the process theoretic notion of stuttering
equivalence in the setting of parity games. We demonstrate
that stuttering equivalent vertices have the same winner
in the parity game. This means that solving a parity game can be
accelerated by minimising the game graph with respect to stuttering
equivalence. While, at the outset, it might not be clear that this
strategy should pay off, our experiments using typical verification
problems illustrate that stuttering equivalence speeds up solving
parity games in many cases.

\end{abstract}

\section{Introduction}

Parity games~\cite{EJ:91,McN:93,Zie:98} are played by two players (called
\emph{even} and \emph{odd}) on a directed graph in which vertices have
been assigned \emph{priorities}.  Every vertex in the graph belongs to
exactly one of these two players. The game is played by moving a token
along the edges in the graph indefinitely; the edge that is moved along
is chosen by the player owning the vertex on which the token currently
resides. Priorities that appear infinitely often along such infinite
plays then determine the winner of the play.

Solving a parity game essentially boils down to computing the set of 
vertices that, if the token is initially placed on a vertex in this 
set, allows player \emph{even} (resp. \emph{odd}) to win. This problem 
is known to be in UP$\,\cap\,$co-UP, a result due to 
Jurdzi\'nski~\cite{Jur:98}; it is still an open problem whether there
is a polynomial time algorithm for the problem, but even in case
such an algorithm is found, it may not be the most efficient algorithm
in practice.

Parity games play a crucial role in verification; the model checking
problem for the modal $\mu$-calculus can be reduced to the problem of
solving a given parity game. It is therefore worthwile to investigate
methods by which these games can be solved efficiently in practice.
In~\cite{FL:09}, Friedman and Lange describe a meta-algorithm that,
combined with a set of heuristics, appears to have a positive impact
on the time required to solve parity games.  Fritz and Wilke consider
more-or-less tried and tested techniques for \emph{minimising} parity
games using novel refinement and equivalence relations, see~\cite{FW:06}.
The delayed simulation they introduce, and its induced equivalence
relation, however, are problematic for quotienting, which is why they go
on to define two variations of delayed simulations that do not suffer
from this problem. As stated in~\cite{Wil:05}, however, ``Experiments
indicate that simplifying parity games using our approach before solving
them is not faster than solving them outright in practice''.

Despite the somewhat unsatisfactory performance of the delayed simulation
in practice, we follow a methodology similar to the one pursued by
Fritz and Wilke. As a basis for our investigations, we consider 
\emph{stuttering equivalence}~\cite{BCG:88}, which originated in the 
setting of Kripke Structures. Stuttering equivalence has two
qualities that make it an interesting candidate for minimising parity
games. Firstly, vertices with the same player and priority are only 
distinguished on the basis of their future branching behaviour,
allowing for a considerable compression. Secondly, stuttering equivalence
has a very attractive worst-case time complexity of $\mathcal{O}(n \cdot
m)$, for $n$ vertices and $m$ edges, which is in stark contrast to the
far less favourable time complexity required for delayed simulation,
which is $\mathcal{O}(n^3\cdot m \cdot d^2)$, where $d$ is the
number of different priorities in the game.  In addition to these,
stuttering equivalence has several other traits that make it appealing:
quotienting is straightforward, distributed algorithms for computing
stuttering equivalence have been developed (see \eg~\cite{BO:03}),
and it admits efficient, scalable implementations using BDD technology \cite{WHH+:06}.

On the basis of the above qualities, stuttering equivalence is likely to
significantly compress parity games that stem from typical model checking
problems.  Such games often have a rather limited number of priorities
(typically at most three), and appear to have regular structures. We note
that, as far as we have been able to trace, quotienting parity games
using stuttering equivalence has never been shown to be sound. Thus,
our contributions in this paper are twofold.

First, we show that stuttering equivalent vertices are won by the same
player in the parity game.  As a side result, given a winning strategy
for a player for a particular vertex, we obtain winning strategies for
all stuttering equivalent vertices.  This is of particular interest
in case one is seeking an explanation for the solution of the game,
for instance as a means for diagnosing a failed verification.

Second, we experimentally show that computing and subsequently
solving the stuttering quotient of a parity game is in many cases
\emph{faster} than solving the original game.  In our comparison,
we included several competitive implementations of algorithms for
solving parity game, including several implementations of \emph{Small
Progress Measures}~\cite{Jur:00} and McNaughton's \emph{recursive}
algorithm~\cite{McN:93}.  Moreover, we also compare it to quotienting
using \emph{strong bisimulation}~\cite{Par:81}.  While we do not claim
that stuttering equivalence minimisation should always be performed
prior to solving a parity game, we are optimistic about its effects in
practical verification tasks.

\paragraph{Structure.} The remainder of this paper is organised
as follows. Section~\ref{sec:preliminaries} briefly introduces the
necessary background for parity games. In Section~\ref{sec:equivalences}
we define both strong bisimilarity and stuttering equivalence
in the setting of parity games; we show that both can be used for
minimising parity games. Section~\ref{sec:experiments} is devoted to
describing our experiments, demonstrating the efficacy of stuttering
equivalence minimisation on a large set of verification problems.
In Section~\ref{sec:conclusions}, we briefly discuss future work and
open issues.

\section{Preliminaries}
\label{sec:preliminaries}

We assume the reader has some familiarity with
parity games; therefore, the  main purpose of this section is to fix
terminology and notation. For an in-depth treatment of these games,
we refer to~\cite{McN:93,Zie:98}.

\subsection{Parity Games}
A parity game is a game played by players \emph{even} (represented by
the symbol $\even$) and \emph{odd} (represented by the symbol $\odd$). 
It is played on a total finite directed graph, the vertices of which can be won
by either $\even$ or $\odd$. The objective of the game is to find the 
partitioning that separates the vertices won by $\even$ from those won by $\odd$.
In the following text we formalise this definition, and we introduce some
concepts that will make it easier to reason about parity games.

\begin{definition}
A parity game $\game$ is a directed graph $(V, \to, \priority, \getplayername)$,
where 
\begin{itemize}
\item $V$ is a finite set of vertices,
\item $\to \subseteq V \times V$ is a total edge relation (\ie, for each
$v\in V$ there is at least one $w \in V$ such that $(v, w) \in \to$),

\item $\Omega: V \to \N$ is a priority function that assigns priorities to vertices,

\item $\getplayername: V \to \{\even, \odd\}$ is a function assigning vertices
to players.

\end{itemize}
\end{definition}
Instead of $(v, w) \in \to$ we will usually write $v \to w$.  Note that,
for the purpose of readability later in this text, our definition
deviates from the conventional definition: instead of requiring a
partitioning of the vertices $V$ in vertices owned by player even 
and player odd, we achieve the same
through the function $\getplayername$. 

\paragraph{Paths.} A sequence of vertices $v_1, \ldots, v_n$ for which
$v_i \to v_{i+1}$ for all $1 \leq i < n$ is called a \emph{path}, and
may be denoted using angular brackets: $\path{v_1, \ldots, v_n}$. The
concatenation $p \pathconcat q$ of paths $p$ and $q$ is again a path. We
use $p_n$ to denote the $n^\textrm{th}$ vertex in a path $p$. The set of
paths of length $n$, for $n \ge 1$ starting in a vertex $v$ is defined
inductively as follows.
\begin{align*}
\prefixes[1]{}{v} & = \{ \path{v} \} \\
\prefixes[n+1]{}{v} & = \{ \path{v_1,\ldots,v_n,v_{n+1}}
~|~ \path{v_1,\ldots,v_n} \in \prefixes[n]{}{v} \land v_n \to v_{n+1} \}
\end{align*}
We use $\prefixes[\omega]{}{v}$ to denote the set of infinite paths
starting in $v$. The set of all paths starting in $v$, both finite and
infinite is defined as follows:
\begin{align*}
\prefixes{}{v} & = \prefixes[\omega]{}{v} \cup \bigcup_{n \in \N} \prefixes[n]{}{v} \\
\end{align*}
\paragraph{Winner.} A game starting in a vertex $v \in V$ is played by
placing a token on $v$, and then moving the token along the edges in the
graph. Moves are taken indefinitely according to the following simple
rule: if the token is on some vertex $v$, player $\getplayer{v}$ moves
the token to some vertex $w$ such that $v \to w$. The result is an
infinite path $p$ in the game graph.  The \emph{parity} of the lowest
priority that occurs infinitely often on $p$ defines the \emph{winner}
of the path. If this priority is even, then player $\even$ wins, otherwise
player $\odd$ wins.

\paragraph{Strategies.} A \emph{strategy} for player $\player$ is
a partial function $\phi : V^{*} \to V$, that for each path ending in a
vertex owned by player $\player$ determines the next vertex to be played
onto. A path $p$ of length $n$ is \emph{consistent} with a strategy $\phi$
for player $\player$, denoted $\consistent{p}{\phi}$, if and only if
for all $1 \leq j < n$ it is the case that $\path{p_1,\ldots,p_j}
\in \dom{\phi}$ and $\getplayer{p_j} = \player$ imply $p_{j+1} =
\phi(\path{p_1,\ldots,p_j})$. The definition of consistency is extended
to infinite paths in the obvious manner. We denote the set of paths that 
are consistent with a given strategy $\phi$, starting in a vertex $v$ by
$\prefixes{\phi}{v}$; formally, we define:
\begin{align*}
\prefixes{\phi}{v} & = \{ p \in \prefixes{}{v} ~|~ \consistent{p}{\phi} \}
\end{align*}
A strategy $\phi$ for player $\player$ is said
to be a \emph{winning strategy} from a vertex $v$ if and only if $\player$
is the winner of every path that starts in $v$ and that is consistent
with $\phi$. It is known from the literature that each vertex in the game
is won by exactly one player; effectively, this induces a partitioning
on the set of vertices $V$ in those vertices won by player $\even$ and
those vertices won by player $\odd$. 

\paragraph{Orderings.} We assume that $V$ is ordered by an arbitrary,
total ordering $\vertexorder$. The minimal element of a non-empty
set $U\subseteq V$ with respect to this ordering is denoted
$\vertexordermin(U)$.  Let $|v,u|$ denote the least number of edges
required to move from vertex $v$ to vertex $u$ in the graph. We define
$|v,u|=\infty$ if $u$ is unreachable from $v$. For each vertex $u \in V$,
we define an ordering $\targetorder{u} \subseteq V \times V$ on vertices,
that intuitively orders vertices based on their proximity to $u$, with a
subjugate role for the vertex ordering $\vertexorder$:
$$
v \targetorder{u} v' \textrm{ iff } |v,u| < |v',u|
\textrm{ or } (|v,u| = |v',u| \textrm{ and } v \vertexorder v')
$$
Observe that $u \targetorder{u} v$ for all $v \not= u$. The minimal
element of $U \subseteq V$ with respect to $\targetorder{u}$ is written
$\targetordermin{u}(U)$.

\section{Strong Bisimilarity and Stuttering Equivalence}
\label{sec:equivalences}

Process theory studies refinement and equivalence
relations, characterising the differences between models of systems that
are observable to entities with different observational powers. Most equivalence
relations have been studied for their computational complexity, giving
rise to effective procedures for deciding these equivalences.  Prominent
equivalences are \emph{strong bisimilarity}, due to Park~\cite{Par:81}
and \emph{stuttering equivalence}~\cite{BCG:88}, proposed by Browne,
Clarke and Grumberg.

Game graphs share many of the traits of the system models studied
in process theory. As such, it is natural to study refinement
and equivalence relations for such graphs, see \eg, \emph{delayed
simulation}~\cite{FW:06}. In the remainder of this section, we recast the
bisimilarity and stuttering equivalence to the setting
of parity games, and show that these are finer than \emph{winner equivalence},
which we define as follows.

\begin{definition} Let $\game = (V, \to, \priority, \getplayername)$ be
a parity game. Two vertices $v, v' \in V$ are said to
be \emph{winner equivalent}, denoted $v \winner v'$ iff
$v$ and $v'$ are won by the same player.
\end{definition}

Because every vertex is won by exactly one player (see \eg~\cite{Zie:98}),
winner equivalence partitions $V$ into a subset won by player $\even$
and a subset won by player $\odd$. Clearly, winner equivalence is therefore 
an equivalence relation on the set of vertices of a given parity game.  
The problem of deciding winner equivalence, is in UP$\,\cap\,$co-UP, 
see~\cite{Jur:98}; all currently known algorithms require time 
exponential in the number of priorities in the game.

We next define strong bisimilarity for parity games; basically, we
interpret the priority function and the partitioning of vertices in players
as state labellings.

\begin{definition}
Let $\game = (V,\to,\priority,\getplayername)$ be a parity game. A
symmetric relation $\R \subseteq V \times V$ is a \emph{strong bisimulation}
relation if $v \R v'$ implies
\begin{itemize}
\item $\priority(v) = \priority(v')$ and $\getplayer{v} = \getplayer{v'}$;
\item for all $w \in V$ such that $v \to w$, there should be a 
$w' \in V$ such that $v' \to w'$ and $w \R w'$.
\end{itemize}
Vertices $v$ and $v'$ are said to be \emph{strongly bisimilar},
denoted $v \sb v'$, iff a strong bisimulation relation $\R$ exists
such that $v \R v'$.

\end{definition}
Strong bisimilarity is an equivalence relation on the vertices
of a parity game; quotienting with respect to strong bisimilarity is
straightforward. It is not hard to show that strong bisimilarity is
strictly finer than winner equivalence.  Moreover, quotienting can be
done effectively with a worst-case time complexity of $\mathcal{O}(|V|
\log |V|)$.  

Strong bisimilarity quotienting prior to solving a parity game can in some
cases be quite competitive. One of the drawbacks of strong bisimilarity,
however, is its sensitivity to counting (in the sense that it will not
identify vertices that require a different number of steps to reach a
next equivalence class), preventing it from compressing the game graph
any further.

Stuttering equivalence shares many of the characteristics
of strong bisimilarity, and deciding it has only a slightly worse
worst-case time complexity. However, it is insensitive to counting,
and is therefore likely to lead to greater reductions. Given these
observations, we hypothesise (and validate this hypothesis in
Section~\ref{sec:experiments}) that stuttering equivalence
outperforms strong bisimilarity and, in most instances, reduces the time
required for deciding winner equivalence in parity games stemming from
verification problems.

We first introduce stuttering bisimilarity~\cite{dNV:95}, a coinductive
alternative to the stuttering equivalence of Browne, Clarke and Grumberg;
we shall use the terms stuttering bisimilarity and stuttering equivalence
interchangeably. The remainder of this section is then devoted to showing
that stuttering bisimilarity is coarser than strong bisimilarity, but
still finer than winner equivalence. The latter result allows one to
pre-process a parity game by quotienting it using stuttering equivalence.


\begin{definition} Let $\game = (V, \to, \priority, \getplayername)$ be
a parity game. Let $\R{} \subseteq V \times V$.
An infinite path $p$ is $\R$-divergent, denoted $\divr{\R}{p}$ iff
$p_1 \R p_i$ for all $i$. Vertex $v \in V$ allows for
divergence, denoted $\divr{\R}{v}$ iff there is a path $p$ such
that $p_1 = v$ and $\divr{\R}{p}$.
\end{definition}
We generalise the transition relation $\to$ to its reflexive-transitive
closure, denoted $\Longrightarrow$, taking a given relation $\R$ on
vertices into account. The generalised transition relation is used
to define stuttering bisimilarity.
Let $\game = (V, \to, \priority, \getplayername)$ be a parity game
and let 
$\R{} \subseteq V \times V$ be a relation on its vertices. Formally, we
define the relations $\to_R \subseteq V \times V$ and $\To_R \subseteq V \times V$ through the following set of deduction rules.
$$
\frac{v \to w \qquad v \R w}{v \to_R w} \qquad
\frac{}{v \To_R v} \qquad
\frac{v \to_R w \qquad w \To_R v'}{v \To_R v'}
$$
We extend this notation to paths: we sometimes write
$\path{v_1,\ldots,v_n} \to u$ if $v_n \to u$; similarly, we write 
$\path{v_1,\ldots,v_n} \to_R u$ and $\path{v_1,\ldots,v_n} \To_R u$.

\begin{definition}\label{def:stut}
  Let $\game = (V,\to,\priority,\getplayername)$ be a parity game. Let $\R \subseteq V
  \times V$ be a symmetric relation on vertices; $\R$ is a \emph{stuttering 
  bisimulation} if $v \R v'$ implies
  \begin{itemize}
    \item $\priority(v) = \priority(v')$ and $\getplayer{v} = \getplayer{v'}$;
    \item $\divr{\R}{v}$ iff $\divr{\R}{v'}$;
    \item If $v \to u$, then either ($v \R u \land u \R v'$), or
          there are $u',w$, such that $v' \To_R w \to u'$ and $v \R w$ and $u
          \R u'$;
  \end{itemize}
  Two states $v$ and $v'$ are said to be \emph{stuttering bisimilar},
  denoted $v \stut
  v'$ iff there is a stuttering bisimulation relation $\R$, such that $v \R v'$.
\end{definition}
Note that stuttering bisimilarity is the largest stuttering
bisimulation.  Moreover, stuttering bisimilarity is an equivalence relation,
see~\eg~\cite{dNV:95,BCG:88}. In addition, quotienting with respect to
stuttering bisimilarity is straightforward.

Stuttering bisimilarity between vertices extends naturally to finite
paths.  Paths of length 1 are equivalent if the vertices they consist of
are equivalent.  If paths $p$ and $q$ are equivalent, then $p \pathconcat
\path{v} \stut q$ iff $v$ is equivalent to the last vertex in $q$
(and analogously for extensions of $q$), and $p \pathconcat \path{v}
\stut q \pathconcat \path{w}$ iff $v \stut w$.  An infinite path $p$ is
equivalent to a (possibly infinite) path $q$ if for all finite prefixes
of $p$ there is an equivalent prefix of $q$ and \emph{vice versa}.

\medskip

We next set out to prove that stuttering bisimilarity is finer than
winner equivalence. Our proof strategy is as follows: given that there is
a strategy $\phi$ for player $\player$ from a vertex $v$, we define a 
strategy for player $\player$ that from vertices equivalent to $v$ schedules 
only paths that are stuttering bisimilar to a path starting in $v$ that is 
consistent with $\phi$.

If after a number of moves a path $p$ has been played, and our strategy
has to choose the next move, then it needs to know which successors for $p$ 
will yield a path for which again there is a stuttering bisimilar path that
is consistent with $\phi$. To this end we introduce the set $\entry{\phi,v}{p}$.

Let $\phi$ be an arbitrary strategy, $v$ an arbitrary vertex owned by the 
player for which $\phi$ defines the strategy, and let $p$ be an arbitrary 
path. We define $\entry{\phi,v}{p}$ as the set of vertices in new classes, 
reachable by traversing $\phi$-consistent paths that start in $v$ and that 
are stuttering bisimilar to $p$.
$$
\entry{\phi,v}{p} =
\{
  u \in V ~|~ \exists q \in \prefixes{\phi}{v}: p \stut q
  \land \consistent{q \pathconcat \path{u}}{\phi} \land q \pathconcat \path{u} \not\stut
  q
\}
$$
Observe that not all vertices in $\entry{\phi,v}{p}$ have to be in the same
equivalence class, because it is not guaranteed that all paths $q \in 
\prefixes{\phi}{v}$, stuttering bisimilar to $p$, are extended by $\phi$ 
towards the same equivalence class.

Suppose the set $\entry{\phi,v}{p}$ is non-empty; in this case, our strategy
should select a \emph{target class} to which $p$ should be extended. Because
stuttering bisimilar vertices can reach the same classes, it does not matter
which class present in $\entry{\phi,v}{p}$ is selected as the target class.
We do however need to make a unique choice; to this end we use the total
ordering $\vertexorder$ on vertices.
$$
\targetclass{\phi,v}{p} = \{ u \in V ~|~ u \stut \vertexordermin(\entry{\phi,v}{p}) \}
$$
Not all vertices in the target class need be reachable from $p$, but there
must exist at least one vertex that is.
We next determine a \emph{target vertex}, by selecting a unique, reachable
vertex from the target class. This
target of $p$, given a strategy $\phi$ and a vertex $v$ is denoted $\target{\phi,v}{p}$; note that
the ordering $\vertexorder$ is again used to uniquely determine a vertex from the set
of reachable vertices.
$$
\target{\phi,v}{p} =
\vertexordermin \{ u \in \targetclass{\phi,v}{p} ~|~ \exists w \in V: p \To_{\stut} w \to u \}
$$
\begin{definition}
We define a strategy $\mimick$ for player $i$ that, given some strategy $\phi$ for player $i$ and a vertex $v$, 
allows only paths to be scheduled that have a stuttering bisimilar
path starting in $v$ that is scheduled by $\phi$. It is defined as follows.
  $$
  \begin{array}{c}
  \mimick(p) = \left\{ 
    \begin{array}{ll}
     \targetordermin{t}\{u \in V ~|~ p \to_{\stut} u \}, 
       \qquad & 
       \begin{array}{l}
         t=\target{\phi,v}{p} \\
         p \not\to \target{\phi,v}{p} \\
         \entry{\phi,v}{p} \not= \emptyset
       \end{array} \medskip \\ 
         \target{\phi,v}{p}
       \qquad &
       \begin{array}{l}
        p \to \target{\phi,v}{p}\\
        \entry{\phi,v}{p} \not= \emptyset 
       \end{array} \medskip \\
     \vertexordermin\{u \in V \,|\, p \rightarrow_{\stut} u \}, & 
       \begin{array}{l}
         \entry{\phi,v}{p} = \emptyset
       \end{array} \medskip \\
    \end{array}
  \right.
  \end{array}
  $$
\end{definition}

\renewcommand{\pathprefix}[3]{\vec{#1}}
\newcommand{\slice}[3]{\pathprefix{#1}{}{}_{#2} \ldots \pathprefix{#1}{}{}_{#3}}

\begin{lemma}
\label{lem:strategy_for_stut}
Let $\phi$ be a strategy for player $i$ in an arbitrary parity game. Assume that $v, w \in V$ and $v \stut w$, and let $\psi=\mimick$. Then
$$
  \forall l \in \N:
  \forall p \in \prefixes[l+1]{\psi}{w}:~
  \exists k \in \N:
  \exists q \in \prefixes[k]{\phi}{v}:~
  p \stut q
$$
\end{lemma}
\begin{proof}
  We proceed by induction on $l$. For $l = 0$, the desired implication follows 
  immediately. 
  For $l=n+1$, assume that we have a path  $p \in \prefixes[n+1]{\psi}{w}$. 
  Clearly, $\path{p_1,\ldots,p_n}$ is also consistent with $\psi$. The induction hypothesis yields us a $q\in\prefixes[k]{\phi}{v}$ for some $k\in\N$ such that $\path{p_1,\ldots,p_n} \stut q$. Let $q$
  be such. We distinguish the following cases:
  \begin{enumerate}
    \item $p_n \stut p_{n+1}$. In this case, clearly $p
    \stut \path{p_1,\ldots,p_n} \stut q$, which finishes this
    case.
    
    \item $p_n \not\stut p_{n+1}$. We again distinguish two cases:
    \begin{enumerate}
      \item Case $\getplayer{p_n} \neq i$. Since
      $p_n \stut q_k$, we find that there must be states $u, w \in V$ such that
      $q_k \To_{\stut} w \to u$ and $p_{n+1} \stut u$. So there must be a path
      $r$ and vertex $u$ such that $p \stut q \pathconcat r \pathconcat \path{u}$, for which
      we know that $r \stut q_k$. Therefore, all vertices in $r$ are owned by
      $\getplayer{q_k}=\getplayer{p_n}$, so $\phi$ is not defined for the extensions of 
      $q$ along $\pi$. We can therefore conclude that $\consistent{q \pathconcat r \pathconcat \path{u}}{\phi}$.

      \item Case $\getplayer{p_n} = i$. Then it must be the case that
      $p_{n+1} = \target{\phi,v}{\path{p_1, \ldots, p_n}}$. By definition, that 
      means that there is a $\phi$-consistent path $r \in
      \prefixes{\phi}{v}$, such that $r \stut p$. \qed
    \end{enumerate}
  \end{enumerate} 
\end{proof}
In the following
lemma we extend the above obtained result to infinite paths.

\begin{lemma}\label{lem:strategy_for_stut_infinite}
Let $\phi$ be a strategy for player $i$ in an arbitrary parity game. Assume that $v, w \in V$ and $v \stut w$, and let $\psi=\mimick$. Then
$$
  \forall p \in \prefixes[\omega]{\psi}{w}:
  \exists q \in \prefixes[\omega]{\phi}{v}:
  p \stut q.
$$
\end{lemma}
\begin{proof}
  Suppose we have an infinite path $p \in \prefixes[\omega]{\psi}{w}$. Using lemma \ref{lem:strategy_for_stut} we can obtain a path $q$ starting in $v$ that is stuttering bisimilar, and that is consistent with $\phi$. The lemma does not guarantee, however, that $q$ is of infinite length. We show that if $q$ is finite, it can always be extended to an infinite path that is still consistent with $\phi$.

  Notice that paths can be partitioned into subsequences of vertices
  from the same equivalence class, and that two stuttering bisimilar
  paths must have the same number of partitions.

  Suppose now that $q$ is of finite length, say $k+1$. Then $p$ must
  contain such a partition that has infinite size. In particular, there
  must be some $n\in\N$ such that $p_{n+j} \stut p_{n+j+1}$ for all $0
  \le j \le |V|$. We distinguish two cases.

  \begin{enumerate}
    \item $\getplayer{p_n} = i$. We show that then also $\entry{\phi,v}{\path{p_0, p_1, \ldots p_n}}
    = \emptyset$. Suppose this is not the case. Then we find that for some $u \in V$, 
    $u = \target{\phi,v}{\pi}$ exists, and therefore $p_{n+j} \targetorder{u} 
    p_{n+j+1}$ for all $j \le |V|$. Since $\targetorder{u}$ is total, this means that 
    the longest chain is of length $|V|$, which contradicts our assumptions. So, 
    necessarily $\entry{\phi,v}{\path{p_0, p_1, \ldots p_n}} = \emptyset$, meaning that no
    path that is consistent with $\phi$ leaves the class of $p_n$. But this means that
    the infinite path that stays in the class of $p_n$ is also consistent with $\phi$.

    \item $\getplayer{p_n} \neq i$. Since $p_n \stut q_k$, also $\getplayer{q_k} \neq i$.
    Since $p_n \stut p_{n+j}$ for all $j \le |V| +1$, this means that
    there is a state $u$, such that $u = p_{n+l} = p_{n+l'}$.  But this
    means that $u$ is divergent. Since $\getplayer{u} \neq i$, and $u \stut q_k$,
    we find that also $q_k$ is divergent. Therefore, there is an infinite
    path with prefix $q$ that is consistent with $\phi$ and that
    is stuttering bisimilar to $p$. \qed
  \end{enumerate}
\end{proof}

\begin{theorem} 
Stuttering bisimilarity is strictly finer than winner equivalence,
\ie, $\stut \subseteq \winner$.
\end{theorem}
\begin{proof}
  The claim follows immediately from lemma \ref{lem:strategy_for_stut_infinite} and the
  fact that two stuttering bisimilar infinite paths have the same infinitely occurring
  priorities. Strictness is immediate. \qed
\end{proof}
Note that strong bisimilarity is strictly finer than stuttering
bisimilarity; as a result, it immediately follows that strong bisimilarity
is finer than winner equivalence, too.

As an aside, we point out that our proof of the above theorem relies
on the construction of the strategy $\mimick{}$; its purpose, however,
exceeds that of the proof. If, by solving the stuttering bisimilar
quotient of a given parity game $\game$, one obtains a winning strategy
$\phi$ for a given player, $\mimick{}$ defines the winning strategies
for that player in $\game$. This is of particular importance in case an
explanation of the solution of the game is required, for instance when the game
encodes a verification problem for which a strategy helps explain the
outcome of the verification (see \eg~\cite{SS:98}). It is not immediately
obvious how a similar feature could be obtained in the setting of, say,
the delayed simulations of Fritz and Wilke~\cite{FW:06}, because vertices
that belong to different players and that have different priorities can be
identified through such simulations.

\section{Experiments}
\label{sec:experiments}

We next study the effect that stuttering equivalence minimisation has in a 
practical setting. We do this by solving parity games that originate from three
different sources (we will explain more later) using three
different methods: direct solving, solving after bisimulation reduction and
solving after stuttering equivalence reduction. Parity games
are solved using a number of different algorithms, \viz a naive C++ implementation of the
\emph{small progress measures} algorithm due to Jurzi{\'n}ski, and the optimized
and unoptimized variants that are implemented in the PGSolver tool \cite{FL:09} 
of the small progress measures algorithm, the recursive algorithm due to 
McNaughton \cite{McN:93}, the bigstep algorithm due to Schewe \cite{Sch:07}
and a strategy improvement algorithm due to V{\"o}ge \cite{VJ:00}.
We compare the time needed by these methods to solve the parity games, and
we compare the sizes of the parity games that are sent to the solving algorithms.

To efficiently compute bisimulation and stuttering equivalence for parity
games we ad{\color{red}a}pted a single-threaded implementation of the corresponding
reduction algorithms by Blom and Orzan~\cite{BO:03} for labelled
transition systems.

All experiments were conducted on a machine consisting of 28 
Intel\textregistered{} Xeon\textregistered{} E5520 Processors running at 
2.27GHz, with 1TB of shared main memory, running a 64-bit Linux distribution 
using kernel version 2.6.27. None of our experiments employ multi-core features.

\subsection{Test sets}

The parity games that were used originate from three different sources. Our
main interest is in the practical implications of stuttering equivalence 
reduction on solving model checking problems, so a number of typical model
checking problems have been selected and encoded into parity games. We describe each of
these problems in a little bit more detail.
\begin{description}
\item[IEEE1394] Five properties of the Firewire Link-Layer protocol (1394) 
  \cite{Lut:97} were considered, as they are described in \cite{SM:98}. They are
  numbered I--V in the order in which they can be found in that document.
\item[Lift] Four properties are checked on the specification of a lift in
  \cite{GPW:03}; a liveness property (I), a property that expresses the absence
  of deadlock (II) and two safety properties (III and IV). These typical model
  checking properties are expressed as alternation-free $\mu$-calculus formulae.
\item[SWP] On a model of the sliding window protocol \cite{BFGPP:05}, a 
  fairness property (I) and a safety property (II) are verified, as well as 7
  other fairness, liveness and safety properties.
\end{description}
Note that some of the properties are described by alternation free
$\mu$-calculus formulae, whereas others alternating.

The second test set was taken from \cite{FL:09} and consists of several 
instances of the elevator problem and the Hanoi towers problem described in
that paper. For the latter, a different encoding was devised and added
to the test set.

Lastly, a number of equivalence checking problems was encoded into parity games
as described in \cite{CPPW:07}.

The parity games induced by the alternation free $\mu$-calculus formulae
have different numbers of priorities, but the priorities along the
paths in the parity games are ascending. In contrast, the paths in the
parity games induced by alternating properties have no such property
and are therefore computationally more challenging. Note that the parity games
generated for these problems only have limited alternations between vertices
owned by player $\even$ and $\odd$ in the paths of the parity games.

The problems taken from \cite{FL:09}, as well as some of the equivalence
checking problems, give rise to parity games with alternations between both
players and priorities.

\subsection{Results}

To analyse the performance of stuttering equivalence reduction, we measured the
number of vertices and the number of edges in the original parity games, the
bisimulation-reduced parity games and the stuttering-reduced parity games. The
results for the IEEE1394, Lift and SWP problems are shown in Table~\ref{tab:sizes}.
For the Elevator model from \cite{FL:09}, the results are shown
in Table~\ref{tab:sizes_pgsolver}.

Figure \ref{fig:comparison}.a compares these sizes graphically; each plot point
represents a parity game, of which the position along the $y$-axis is determined
by its stuttering-reduced size, and the position along the $x$-axis by its 
original size and its bisimulation-reduced size, respectively. The plotted sizes
are the sum of the number of vertices and the number of edges.

In addition to these results, we measured the time needed to reduce and to solve
the parity games. The time needed to solve a parity game using stuttering 
equivalence reduction is computed as the time needed to reduce the parity game
using stuttering equivalence, plus the time needed by \emph{fastest} of the 
solving algorithms to solve the reduced game. A similar measure was recorded for
solving parity games using bisimulation reduction. Also, the time needed to 
solve these games directly was measured. The results are plotted in figure
\ref{fig:comparison}.b. Again, every data point is a parity game, of which the
solving times determine the position in the scatter plot.

\begin{table}[!ht]
  \centering
  \vspace{1em}
  \caption{Statistics for the parity games for the 1394, Lift and SWP
  experiments. In the Lift case, $N$ denotes the number of distributed
  lifts; in the case of SWP, $N$ denotes the size of the window. The
  number of priorities in the original (and minimised)
  parity games is listed under Priorities.}
  \label{tab:sizes}  
  \setlength{\tabcolsep}{3.5pt}
\begin{tabular}{lll||rr|rr|rr}\\
\textbf{IEEE 1394} & & &
\multicolumn{2}{c}{\textbf{original}} & \multicolumn{2}{c}{$\stut$} & \multicolumn{2}{c}{$\sb$} \\
\hline
\hline
& & & \\
\textbf{Property} & \textbf{Priorities} &  &
\textbf{$|V|$} & \textbf{$|{\to}|$~} & \textbf{$|V|$} & \textbf{$|{\to}|$~} & \textbf{$|V|$} & \textbf{$|{\to}|$} \\
\hline 
 I & 1 &  & 346\,173 & 722\,422 & 1 & 1 & 1 & 1 \\ \hline 
 II & 1 &  & 377\,027 & 679\,157 & 3\,730 & 3\,086 & 5\,990 & 11\,180 \\ \hline 
 III & 4 &  & 1\,179\,770 & 1\,983\,185 & 102 & 334 & 13\,551 & 22\,166 \\ \hline 
 IV & 2 &  & 524\,968 & 875\,296 & 4 & 6 & 10\,814 & 17\,590 \\ \hline 
 V & 1 &  & 1\,295\,249 & 2\,150\,590 & 1 & 1 & 1 & 1  \\ \hline 

\\
\\
\textbf{Lift} & & &
\multicolumn{2}{c}{\textbf{original}} & \multicolumn{2}{c}{$\stut$} & \multicolumn{2}{c}{$\sb$} \\
\hline
\hline
& & & \\
\textbf{Property} & \textbf{Priorities} & \textbf{N} &
\textbf{$|V|$} & \textbf{$|{\to}|$~} & \textbf{$|V|$} & \textbf{$|{\to}|$~} & \textbf{$|V|$} & \textbf{$|{\to}|$} \\
\hline 
I & 2 & 4 & 1\,691 & 4\,825 & 22 & 58 & 333 & 1\,021\\ \hline 
I & 3 & 4 & 63\,907 & 240\,612 & 131 & 450 & 5\,148 & 23\,703 \\ \hline 
I & 4 & 4 & 1\,997\,579 & 9\,752\,561 & 929 & 4\,006 & 74\,059 & 462\,713 \\ \hline
II & 2 & 2 & 846 & 2\,172 & 5 & 9 & 94 & 240 \\ \hline 
II & 2 & 3 & 31\,954 & 121\,625 & 16 & 39 & 1\,092 & 4\,514 \\ \hline 
II & 2 & 4 & 998\,790 & 5\,412\,890 & 64 & 193 & 14\,353 & 80\,043 \\ \hline 
III & 1 & 2 & 763 & 1\,903 & 1 & 1 & 1 & 1 \\ \hline 
III & 1 & 3 & 26\,996 & 99\,348 & 1 & 1 & 1 & 1 \\ \hline 
III & 1 & 4 & 788\,879 & 4\,146\,139 & 1 & 1 & 1 & 1 \\ \hline 
IV & 2 & 2 & 486 & 1\,126 & 4 & 6 & 151 & 396 \\ \hline 
IV & 2 & 3 & 11\,977 & 39\,577 & 5 & 9 & 1\,741 & 6\,951 \\ \hline 
IV & 2 & 4 & 267\,378 & 1\,257\,302 & 7 & 15 & 23\,526 & 122\,230 \\ \hline 

\\
\\
\textbf{SWP} & & &
\multicolumn{2}{c}{\textbf{original}} & \multicolumn{2}{c}{$\stut$} & \multicolumn{2}{c}{$\sb$} \\
\hline
\hline
& & & \\
\textbf{Property} & \textbf{Priorities} & \textbf{N} &
\textbf{$|V|$} & \textbf{$|{\to}|$~} & \textbf{$|V|$} & \textbf{$|{\to}|$~} & \textbf{$|V|$} & \textbf{$|{\to}|$} \\
\hline 

I &3 & 1 &   1\,250 & 3\,391 & 4 & 7 & 314 & 849 \\ \hline 
I &3 & 2 &   14\,882 & 47\,387 & 4 & 7 & 1\,322 & 4\,127 \\ \hline
I &3 & 3 &   84\,866 & 291\,879 & 4 & 7 & 4\,190 & 14\,153 \\ \hline 
I &3 & 4 &   346\,562 & 1\,246\,803 & 4 & 7 & 11\,414 & 40\,557 \\ \hline 
II&2 & 1  &  1\,370 & 4\,714 & 5 & 8 & 90 & 316 \\ \hline  
II&2 & 2  &  54\,322 & 203\,914 & 5 & 8 & 848 & 3\,789 \\ \hline 
II&2 & 3  &  944\,090 & 3\,685\,946 & 5 & 8 & 5\,704 & 28\,606 \\ \hline 
II&2 & 4  &  11\,488\,274 & 45\,840\,722 & 5 & 8 & 34\,359 & 183\,895 \\
\hline
\end{tabular}

\end{table}

\begin{table}[!ht]
\caption{Statistics for the parity games for the FIFO and LIFO Elevator models
taken from \cite{FL:09}. \textbf{Floors} indicates the number of floors.}
\label{tab:sizes_pgsolver}
\vspace{-1em}
\setlength{\tabcolsep}{3.5pt}
\begin{tabular}{lll||rr|rr|rr}\\
\multicolumn{3}{l}{\textbf{Elevator Models}}   &
\multicolumn{2}{c}{\textbf{original}} & \multicolumn{2}{c}{$\stut$} & \multicolumn{2}{c}{$\sb$} \\
\hline
\hline
& & & \\
\textbf{Model} & \textbf{Floors} & \textbf{Priorities} & 
\textbf{$|V|$} & \textbf{$|{\to}|$~} & \textbf{$|V|$} & \textbf{$|{\to}|$~} & \textbf{$|V|$} & \textbf{$|{\to}|$} \\
\hline 
FIFO & 3 & 3 & 564 & 950 & 351 & 661 & 403 & 713 \\ \hline
FIFO & 4 & 3 & 2\,688 & 4\,544 & 1\,588 & 2\,988 & 1\,823 & 3\,223 \\\hline
FIFO & 5 & 3 & 15\,684 & 26\,354 & 9\,077 & 16\,989 & 10\,423 & 18\,335 
\\\hline
FIFO & 6 & 3 &  108\,336 & 180\,898 & 62\,280 & 116\,044 & 71\,563 & 125\,327 \\\hline
FIFO & 7 & 3 & 861\,780 & 1\,431\,610 & 495\,061 & 919\,985 & 569\,203 & 994\,127 \\\hline
LIFO & 3 & 3 & 588 & 1\,096 & 326 & 695 & 363 & 732 \\\hline
LIFO & 4 & 3 & 2\,832 & 5\,924 & 866 & 2\,054 & 963 & 2\,151 \\\hline
LIFO & 5 & 3 & 16\,356 & 38\,194 & 2\,162 & 5\,609 & 2\,403 & 5\,850 \\\hline
LIFO & 6 & 3 & 111\,456 & 287\,964 & 5\,186 & 14\,540 & 5\,763 & 15\,117 \\\hline
LIFO & 7 & 3 & 876\,780 & 2\,484\,252 & 16\,706 & 51\,637 & 18\,563 &
53494 \\\hline
\end{tabular}
\end{table}

\begin{figure}[h!]
  \hspace{-0.5em}
  \begin{tabular}{cc}
    \multicolumn{2}{c}{\textbf{(a) Parity game sizes}}\medskip
  \\
    \begin{tikzpicture}[mark size=2.5pt,remember picture,scale=0.77,baseline]
      \begin{loglogaxis}[axis x line=bottom,
                         axis y line=right,
                         xmin=1,xmax=50000000,
                         ymin=1,ymax=5000000,
                         xlabel={Original},
                         ylabel={Stuttering}]
        \addplot[only marks, mark=+] plot file {size_modelchecking_os.txt};
        \addplot[only marks, mark=o] plot file {size_pgsolver_os.txt};
        \addplot[only marks, mark=triangle] plot file {size_equivalencechecking_os.txt};
        \addplot[dotted] plot coordinates {(1, 1) (50000000, 50000000)};
      \end{loglogaxis}
    \end{tikzpicture}
  & \hspace{-1.5em}
    \begin{tikzpicture}[mark size=2.5pt,remember picture,scale=0.77,baseline]
      \begin{loglogaxis}[axis x line=bottom,
                         axis y line=left,
                         xmin=1,xmax=5000000,
                         ymin=1,ymax=5000000,
                         xlabel={Bisimulation},
                         yticklabels={}]
        \addplot[only marks, mark=+] plot file {size_modelchecking_bs.txt};
        \addplot[only marks, mark=o] plot file {size_pgsolver_bs.txt};
        \addplot[only marks, mark=triangle] plot file {size_equivalencechecking_bs.txt};
        \addplot[dotted] plot coordinates {(1, 1) (5000000, 5000000)};
      \end{loglogaxis}
    \end{tikzpicture}
    \bigskip
  \\
    \multicolumn{2}{c}{\textbf{(b) Solving times}\medskip}
  \\
    \begin{tikzpicture}[mark size=2.5pt,remember picture,scale=0.77,baseline]
      \begin{loglogaxis}[axis x line=bottom,
                         axis y line=right,
                         xmin=0.01,xmax=1990,
                         ymin=0.01,ymax=1990,
                         xlabel={Original},
                         ylabel={Stuttering},
                         anchor=south]
        \addplot[only marks, mark=+] plot file {experiments_modelchecking_os.txt};
        \addplot[only marks, mark=o] plot file {experiments_pgsolver_os.txt};
        \addplot[only marks, mark=triangle] plot file {experiments_equivalencechecking_os.txt};
        \addplot[dotted] plot coordinates {(0.01, 0.01) (1990, 1990)};
      \end{loglogaxis}
    \end{tikzpicture}
  & \hspace{-1.5em}
    \begin{tikzpicture}[mark size=2.5pt,remember picture,scale=0.77,baseline]
      \tikzstyle{plot legend}=[
         rounded corners=2.5pt,inner xsep=3pt,inner ysep=2pt,
         draw=black!50,fill=white,
         font=\footnotesize,cells={anchor=center},
         nodes={inner sep=2pt,text depth=0.15em,rounded corners=0pt,right}]
      \pgfplotsset{every axis legend/.style={plot legend,at={(0.95,0.05)},above left}}
      \begin{loglogaxis}[axis x line=bottom,
                         axis y line=left,
                         xmin=0.01,xmax=1990,
                         ymin=0.01,ymax=1990,
                         xlabel={Bisimulation},
                         yticklabels={},
                         anchor=south]
        \addplot[only marks, mark=+] plot file {experiments_modelchecking_bs.txt};
        \addlegendentry{IEEE1394, Lift, SWP};
        \addplot[only marks, mark=o] plot file {experiments_pgsolver_bs.txt};
        \addlegendentry{Hanoi, Elevator};
        \addplot[only marks, mark=triangle] plot file {experiments_equivalencechecking_bs.txt};
        \addlegendentry{Equivalence checking};
        \addplot[dotted] plot coordinates {(0.01, 0.01) (1990, 1990)};
      \end{loglogaxis}
    \end{tikzpicture}
  \end{tabular}
  \caption{Sizes and solving times (in seconds) of the stuttering-reduced parity 
  games set out against sizes and solving times of the original games and of the 
  bisimulation-reduced games. The vertical axis is shared between the plots in
  each subfigure.}
  \label{fig:comparison}
\end{figure}
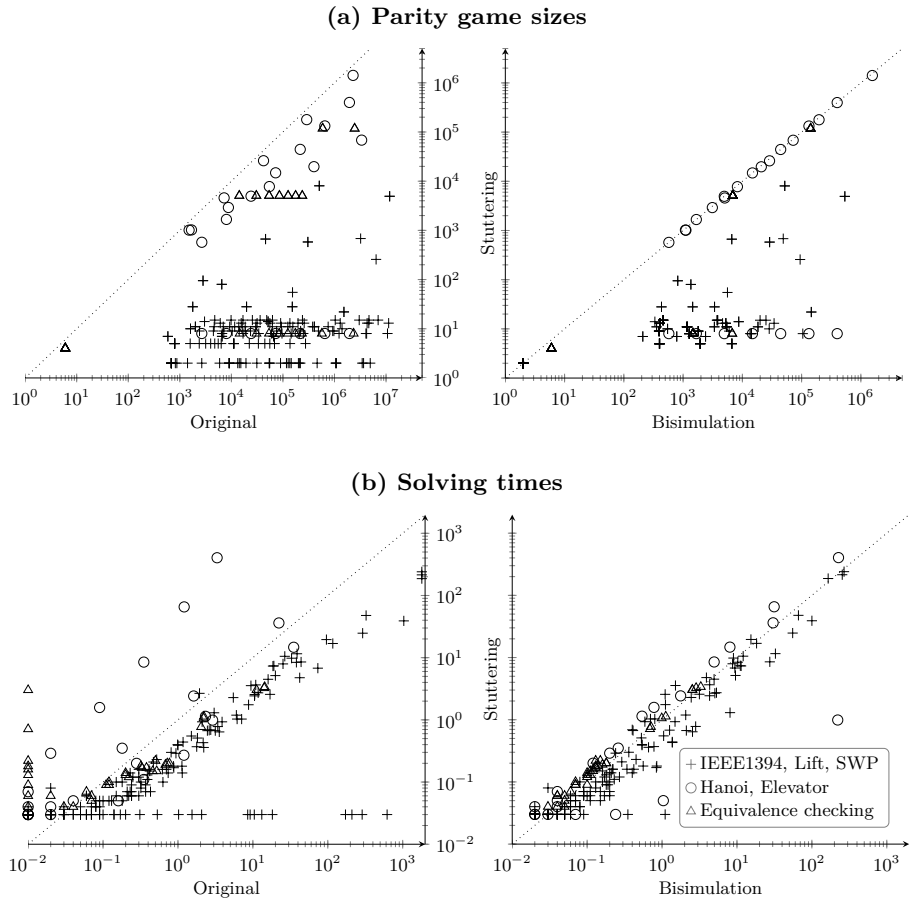

\subsection{Discussion}

At a glance, stuttering reduction seems a big improvement on bisimulation 
reduction in terms of size reduction. Figure \ref{fig:comparison}.a shows 
clearly that stuttering equivalence gives a better size reduction than 
bisimulation equivalence in the majority of cases. The difference is often 
somewhere between a factor ten and a factor thousand. Looking at solving times, 
the results also seem promising. In figure \ref{fig:comparison}.b we see that 
in most cases reducing the game and then solving it costs significantly less 
time. We will discuss the results in more detail for each test set separately.

\subsubsection{IEEE1394, Lift, SWP}

For these cases, we see that the size reduction is always better than that of
bisimulation reduction, unless bisimulation already compressed the parity game
to a single state. Solving times using stuttering equivalence are in general 
better than those of direct solving.

The experiments indicate
that minimising parity games using stuttering equivalence before solving
the reduced parity games is at least as fast as directly solving the
original games.

The second observation we make is that stuttering equivalence reduces the size
quite well for this test set, when compared to the other sets. This may be 
explained by the way in which the parity games were generated. As they encode
a $\mu$-calculus formula together with a state space, repetitive and 
deterministic parts of the state space are likely to generate fragments within
the parity game that can be easily compressed using stuttering reduction.

Lastly, we observe that solving times using bisimulation reduction are not
in general much worse than those using stuttering reduction. The explanation is
simple: both reductions compress the original parity game to such an extent that
the resulting game is small enough for the solvers to solve it in less than a 
tenth of a second.

\subsubsection{Equivalence checking}

The results for these experiments indicate that reduction using stuttering
equivalence sometimes performs poorly. The subset where performance is
especially poor is an encoding of branching bisimilarity, which gives rise to
parity games with alternations both between different priorities as well as
different players. As a result, little reduction is possible.

\subsubsection{Hanoi, Elevator}

Both stuttering equivalence and strong bisimulation reduction perform
poorly on a reachability property for the Hanoi towers experiment, with
the reduction times vastly exceeding the times required for solving the
parity games directly.  A closer inspection reveals that this is caused
by an unfortunate choice for a new priority for vertices induced by a
fixpoint-free subformula.  As a result, all paths in the parity game
have alternating priorities with very short stretches of the same priorities, 
because of which hardly any reduction is possible. We included an encoding of
the same problem which does not contain the unfortunate choice, and indeed
observe that in that case stuttering equivalence does speed up the solving
process.

The LIFO Elevator problem shows results similar to those of the other model 
checking problems. The performance with respect to the FIFO Elevator however is
rather poor. This seems to be due to three main factors: the relatively large 
number of alternating fixed point signs, the alternations between vertices owned
by player $\even$ and vertices owned by player $\odd$, and the low average 
branching degree in the parity game. This indicates that for alternating 
$\mu$-calculus formulae with nested conjunctive and disjunctive subformulae, 
stuttering equivalence reduction generally performs suboptimal. This should not 
come as a surprise, as stuttering equivalence only allows one to compress
sequences of vertices with equal priorities and owned by the same player.

\section{Conclusions}
\label{sec:conclusions}
We have adapted the notion of stuttering bisimilarity to the setting
of parity games, and proven that this equivalence relation can be
safely used to minimise a parity game before solving the reduced
game.

Experiments were conducted to investigate the effect of quotienting
stuttering bisimilarity on parity games originating from model checking
problems. In many practical cases this reduction leads to an improvement
in solving time, however in cases where the parity games involved have
many alternations between odd and even vertices, stuttering bisimilarity
reduction performs only marginally better than strong bisimilarity
reduction. Although we did compare our techniques against a number of
competitive parity game solvers, using other solving algorithms, or even other
implementations of the same algorithms, might give slightly different results,
also depending on heuristics that are implemented for \eg the small progress
measurees algorithm.

The fact that stuttering bisimilarity does not deal at all well with such
alternations leads us to believe that weaker notions of bisimilarity,
in which vertices with different players can be related under certain
circumstances, may resolve the most severe performance problems that
we saw in our experiments.  We regard the investigation of such weaker
relations as future work.

Stuttering bisimilarity has been previously studied in a distributed
setting~\cite{BO:03}. It would be interesting to compare its performance
to a distributed implementation of the known solving algorithms for
parity games. However, we are only aware of a multi-core implementation
of the \emph{Small Progress Measures} algorithm~\cite{PW:08}.

\bibliographystyle{plain}
\bibliography{lit}

\end{document}